\documentclass[twoside]{article}

\usepackage[accepted]{aistats2026}
\usepackage[round]{natbib}

\newcommand{\D}{\mathcal{D}}
\newcommand{\revision}{\textcolor{black}}

\usepackage{amsfonts}
\usepackage{amssymb}
\usepackage{dsfont}
\usepackage{amsthm}
\usepackage{tikz}
\usepackage{hyperref}
\usepackage[capitalise]{cleveref}
\usepackage{algorithm}
\usepackage{algpseudocode}
\usepackage{fontawesome5}
\usepackage{subcaption}
\usepackage{mdframed}
\usepackage{enumitem}

\usetikzlibrary{positioning}

\newtheorem{theorem}{Theorem}
\newtheorem{lemma}{Lemma}

\newtheorem{definition}{Definition}
\newtheorem{assumption}{Assumption}

\DeclareMathOperator*{\argmin}{argmin}
\DeclareMathOperator*{\argmax}{argmax}
\usepackage[dvipsnames]{xcolor}

\newcommand{\old}[1]{}
\usepackage{mathtools}   
\usepackage{autonum}  
\usepackage{bm}
\usepackage{bbm}

\crefname{equation}{}{}
\crefname{assumption}{Assumption}{Assumptions}
\Crefname{appendix}{Appendix}{Appendices}
\begin{document}

%
\runningtitle{Eliciting Truthful Feedback for Preference-Based Learning}

%
\runningauthor{Landolt, Maddux, Schlaginhaufen, Vaishampayan, Kamgarpour}

\twocolumn[
\aistatstitle{Eliciting Truthful Feedback for Preference-Based Learning\\ via the VCG Mechanism}
\aistatsauthor{Leo Landolt\textsuperscript{1,2}, Anna Maddux*\textsuperscript{,1}, Andreas Schlaginhaufen*\textsuperscript{,1},}
\aistatsauthor{Saurabh Vaishampayan\textsuperscript{1}, Maryam Kamgarpour\textsuperscript{1}}\vspace{0.2cm}
\aistatsaddress{\textsuperscript{1}EPFL, Switzerland\quad\textsuperscript{2}ETH Zürich, Switzerland}]
\begin{abstract}

We study resource allocation problems in which a central planner allocates resources among strategic agents with private cost functions in order to minimize a social cost, defined as an aggregate of the agents’ costs. This setting poses two main challenges: (i) the agents’ cost functions may be unknown to them or difficult to specify explicitly, and (ii) agents may misreport their costs strategically. To address these challenges, we propose an algorithm that combines preference-based learning with Vickrey–Clarke–Groves (VCG) payments to incentivize truthful reporting. Our algorithm selects informative preference queries via D-optimal design, estimates cost parameters through maximum likelihood, and computes VCG allocations and payments based on these estimates. In a one-shot setting, we prove that the mechanism is approximately truthful, individually rational, and efficient up to an error of $\tilde{\mathcal O}(K^{-1/2})$ for $K$ preference queries per agent. In an online setting, these guarantees hold asymptotically with sublinear regret at a rate of $\tilde{\mathcal O}(T^{2/3})$ after $T$ rounds. Finally, we validate our approach through a numerical case study on demand response in local electricity markets.
\end{abstract}

\section{INTRODUCTION}
Allocating resources among self-interested agents is a fundamental problem in many real-world systems, including electricity markets, communication networks, and transportation systems \citep{chremos2024mechanism}. In these problems, a central planner aims to determine a resource allocation that minimizes the social cost, defined as an aggregate function of individual agents' costs. Agents' cost functions are typically private, thus, the central planner relies on the agents to communicate their cost functions. However, this is often problematic, either because agents cannot explicitly specify their cost functions or because doing so is difficult.

A prominent example arises in local electricity markets, where a grid operator procures energy flexibility from domestic consumers through deferrable appliance usage, such as shifting or adjusting heating or laundry schedules \citep{li2020transactive,tsaousoglou2022market}. Reporting the exact economic costs for these adjustments can be difficult and impractical for consumers. Instead, they can typically more easily express their preferences over different flexibility options.

This motivates us to consider a setting where agents provide preference feedback to the central planner rather than reporting their cost functions directly. The central planner can then leverage this feedback to infer agents’ underlying cost structures and compute an allocation that minimizes the resulting social cost.

A central challenge, however, is that agents may behave strategically, misreporting their preferences to manipulate the resource allocation outcome for their benefit. This can lead to socially suboptimal outcomes. In electricity markets, for instance, consumers may misreport their preferences to secure more favorable energy tariffs, thereby undermining efficient allocation \citep{yazdani2017strategic}.

In this paper, we investigate how the central planner can elicit truthful preference feedback from agents to solve the resource allocation problem, minimizing the social cost. 

\subsection{Related work}
Learning from preference feedback was first studied in the bandit literature, where the usual numerical rewards were replaced by pairwise preferences \citep{yue2009interactively, ailon2014reducing}. More recently, preference-based learning has received considerable attention across a variety of applications, most prominently in fine-tuning large language models \citep{ziegler2019fine, rafailov2023direct}, reinforcement learning \citep{christiano2017deep}, robotics and human-robot interaction, and personalized healthcare decision support. Preference feedback is especially valuable in settings involving human interaction, as human users are often better at expressing relative judgment in the form of preferences between outcomes than providing numerical evaluations quantifying the value of an outcome \citep{pereira2019online,lee2023aligning}.

A common modeling assumption in this literature is that agents' preferences can be represented by cardinal models, such as the Bradley–Terry model \citep{bradley1952rank, luce1959individual}. Prior works have largely focused on the case where agents truthfully report preferences according to such models \citep{azar2024general,chowdhury2024differentially}. An important open challenge is to understand how preference learning methods perform in settings where agents may strategically misreport their preferences in pursuit of personal gain.

Strategic behavior in multi-agent resource allocation problems has been extensively studied in mechanism design. In settings with transferable utility -- where agents' utilities can be expressed in monetary terms -- mechanisms commonly incorporate payments to incentivize socially desirable and efficient outcomes \citep{fallah2024fair,gois2025performative}. In such cases, an agent’s utility is defined as the payment received minus the incurred costs. A classical result establishes that under the Vickrey--Clarke--Groves (VCG) mechanism, agents are incentivized to truthfully report their cost functions \citep{vickrey1961counterspeculation,clarke1971multipart,groves1973incentives}.

However, this approach requires that agents disclose their entire cost functions explicitly. A more realistic assumption is that agents instead provide numerical feedback, i.e., reporting their realized cost for a given allocation \citep{babaioff2009characterizing,devanur2009price}. Closely related to our setting, \citet{kandasamy2023vcg} study numerical feedback over a finite allocation space, and show that the VCG mechanism incentivizes truthful reporting asymptotically. This is restrictive -- for example in electricity markets, energy flexibility is naturally continuous, and it is unrealistic to expect consumers to report precise economic cost values. \revision{Recent work has also explored auction-based mechanisms for eliciting preference feedback in the context of fine-tuning large language models, but focuses on second-price auctions rather than general resource allocation \citep{zhang2024vickreyfeedback}.} To address these limitations, we study resource allocation problems with a compact allocation space in which agents provide preference feedback rather than exact numerical reports.

\subsection{Contributions} To the best of our knowledge, this is the first work to study preference learning in resource allocation problems with strategic agents under transferable utilities, where agents’ preferences are modeled using the Bradley–Terry framework. Our contributions are threefold:
\begin{itemize}
    \item We propose a resource allocation algorithm that learns agents’ private costs, modeled as a linear function of features, from strategic pairwise preferences. The algorithm uses maximum likelihood estimation for parameter estimation, D-optimal design for query selection, and VCG payments to ensure truthful feedback.
    \item We establish guarantees analogous to the classical VCG mechanism for truthfulness, individual rationality, and efficiency. Our algorithm satisfies these properties approximately in the one-shot setting (\cref{th:approximate}) and asymptotically in the multi-round setting (\cref{th:asymptotic}).
    \item We demonstrate the applicability of our setup to local electricity markets and validate our guarantees through a numerical case study in this domain.
\end{itemize} 

\subsection{Notation} We write $\|{\cdot}\|$ for the Euclidean norm and $\langle {\cdot},{\cdot}\rangle$ for the standard inner product in $\mathbb R^d$. We use the standard notation $\mathcal O(T)$ for asymptotic upper bounds, as well as $\tilde{\mathcal O}(T) = O(T\, \text{polylog}(T))$ for suppressing polylogarithmic terms. Finally, we denote $[N]:=\{1,\dots,N\}$.

\section{BACKGROUND AND PROBLEM FORMULATION}
\label{sec:setup}
We consider a resource allocation problem with $N$ strategic agents, each with a compact allocation set $\mathcal A_i$. A central planner chooses an allocation $a=
(a_1,\dots,a_N)\in\mathcal A$, where $\mathcal A=\mathcal A_1 \times\hdots\times\mathcal A_N$. In our electricity market application, $N$ is the number of consumers and $a_i$ is the energy supplied (kWh).  The allocation results in a continuous cost $c_i:\mathcal A\to\mathbb R_+$ for each agent. The goal of the central planner is to find a feasible allocation that minimizes the social cost
\begin{align}\label{eq:social-cost}
    J(a):=\sum_{i=1}^N c_i(a),\quad \text{s.t. }a\in\mathcal F,
\end{align}
where $\mathcal F\subseteq \mathcal A$ is a compact feasible set. In the electricity market application, for example, $\mathcal F=\{a\in\mathcal A \mid \sum_{i=1}^N a_i=P\}$, where $P$ is the total energy flexibility required by the grid (kWh). This leads to a classical problem in mechanism design: a central planner must elicit private cost functions from strategic agents in order to find a socially efficient allocation.  

\subsection{Auctions}
A common approach to address the above problem are auctions, where each agent submits a bid function $b_i:\mathcal A\to\mathbb R_+$ intended to reflect their cost. A mechanism consists of an allocation rule $a^*:b\to\mathcal A$ and a payment rule $p:b\to\mathbb R^N$ based on the agents' bid functions $b=(b_1,\ldots,b_N)$. Then, the central planner finds an allocation $a^*(b)$ such that:
\[
    a^*(b)\in\arg\min_{a\in\mathcal F} \hat{J}(a;b),
\]
where $\hat{J}(a;b)=\sum_{i=1}^N b_i(a)$. The central planner assigns this allocation $a^*(b)=(a_1^*(b),\dots,a_N^*(b))$ to the agents who receive payments $p(a^*(b))=(p_1(a^*(b)),\dots,p_N(a^*(b)))$. As a result, agent $i$ experiences utility:

\begin{equation}
\label{eq:utility}
  u_i(a^*(b))=p_i(a^*(b))- c_i(a^*(b)).  
\end{equation}

An agent is said to bid truthfully if $\bar b_i(a)=c_i(a)$ for all $a\in\mathcal A$. 
However, agents may misreport their costs and submit an arbitrary bid $b_i$ to increase their utility. For example, in a pay-as-bid mechanism, also known as first-price auction, since payments to agents are equal to their bids, a rational agent would overbid to ensure profit \citep[Chapter 14.2]{karlin2017game}. Consequently, $a^*(b)$ may not minimize the true social cost \cref{eq:social-cost}. Thus, to find a socially efficient allocation, the central planner must carefully design the payment rule in order to elicit truthful bids.

In particular, any desirable mechanism should satisfy the following three properties:
\begin{enumerate}[label={\arabic*)}]
    \item \textbf{Truthfulness:} Truthful bidding with $\bar b_i(a)=c_i(a)$ for all $a\in\mathcal A$, is a weakly dominant strategy Nash equilibrium, that is: $$u_i(a^*(\bar b_i, b_{-i}))\geq u_i(a^*(b_i,b_{-i})),$$ for all $b_i:\mathcal A\to\mathbb R_+$ and $i\in[N]$. We use $b_{-i}$ for the (potentially untruthful) bids of other agents.
    \item \textbf{Individual rationality:} A truthful agent has non-negative utility, that is:
    $$u_i(a^*(\bar b_i,b_{-i}))\geq 0,$$ for all bids $b_{-i}$ of other agents.
    \item \textbf{Efficiency:} When all agents bid truthfully, the resulting allocation minimizes the social cost, that is:\looseness-1 $$a^*(\bar b)\in\argmin_{a\in\mathcal F} J(a).$$
\end{enumerate}
A well-known mechanism that satisfies these properties is the VCG mechanism \citep{vickrey1961counterspeculation,clarke1971multipart,groves1973incentives}, which we discuss in \cref{sec:core-components}.

\subsection{Preference model}
    In our setting, agents cannot directly report their cost functions, but instead express their preferences between two allocations. We model these preferences using the stochastic Bradley–Terry model, which reflects bounded rationality in human decision-making \citep{bradley1952rank, luce1959individual}. 
	\begin{definition}[Bradley--Terry model]
	The probability that agent $i$ prefers allocation $a\in\mathcal A$ over allocation $a'\in\mathcal A$, denoted as $a\succ a'$, is:
    \begin{equation}
        \label{eq:bradley-terry}
        \Pr(a\succ a')=\sigma\left(u_i(a)-u_i(a')\right),     
     \end{equation}
	where $\sigma(x)=1/(1+e^{-x})$ is the sigmoid function.\footnote{The Bradley--Terry model can be extended with a rationality parameter $\beta_i > 0$. Then, the preference model for agent $i$ is $\Pr(a \succ a') = \sigma\left(\beta_i [u_i(a) - u_i(a')]\right)$. Large $\beta_i$ implies more deterministic decisions, whereas small $\beta_i$ introduces more randomness. We assume $\beta_i$ is known and, without loss of generality, set $\beta_i=1$ for all $i$. We discuss this further in \cref{app:rationality}.}
	\end{definition}
    
Applying the Bradley--Terry model to an agent's utility given in Equation \cref{eq:utility}, we obtain:
	\[\Pr(a\succ a')=\sigma\left(p_i(a)-c_i(a)-p_i(a')+c_i(a')\right).\]
    
	The payments $p_i(a)$ and $p_i(a')$ are known, therefore, they only induce an affine shift in the argument of the sigmoid function. We say that preference feedback is truthful if agent $i$ expresses their preference for $a$ over $a'$ according to a binary label $\bar y_i=\mathds{1}(a\succ a')$ where
\begin{align}
    \bar y_i\sim\text{Bernoulli}\bigl(\sigma(p_i(a)-c_i(a)-p_i(a')+c_i(a'))\bigr),
\end{align}
is a Bernoulli-distributed random variable sampled according to the Bradley--Terry model. While $\bar y_i$ represents the agent's truthful preference, they may act strategically to increase their utility. That is, for given allocations $a$ and $a'$, agent $i$ may report a preference $y_i\in\{0,1\}$ sampled from an arbitrary distribution. \revision{Note that we do not impose any structural assumptions on the agents' strategic feedback. In particular, agents can follow dynamic strategies when generating their untruthful preferences.}

\subsection{Problem formulation}
The central planner's goal is to find an allocation that minimizes the social cost \cref{eq:social-cost}. Unlike standard auctions, where agents submit bid functions, the planner actively queries each agent with a pair of allocations $a, a' \in \mathcal{A}$. Agents then report their preference in form of a binary label $y_i=\mathds{1}(a\succ a')$, and the planner learns their costs through repeated interaction. Importantly, to find a socially efficient allocation, the central planner must learn agents' true costs, which is possible only if agents report their preferences truthfully.
\vspace{0.2cm}
\begin{mdframed}
\textbf{Problem:} Can we design preference queries and payment rules to ensure \textit{truthfulness, individual rationality, and efficiency}?
\end{mdframed}
Motivated by the success of the VCG mechanism in auctions, we adopt VCG as our payment rule. We then study the above problem in two regimes. As a warm-up, we analyze a one-shot setting: the planner collects a fixed batch of pairwise preferences (e.g., via a questionnaire), estimates costs, and then computes allocations and payments from that estimate. In Section~\ref{sec:multi-round}, we turn to a multi-round setting in which the planner alternates between exploration and exploitation, improving allocations and payments over time as the cost estimates become more accurate. The one-shot game captures one-time allocations (e.g., crowdsourcing a single task), whereas the multi-round game captures repeated allocations, such as our electricity market application in Section~\ref{sec:experiments}.

\section{ONE-SHOT GAME: ALGORITHM AND ANALYSIS}
\label{sec:one-shot}
In the one-shot game, the planner first collects a set of $K$ pairwise preferences to learn the agents' costs, and then computes the allocation and VCG payments from this estimate. To efficiently explore the possibly infinite set of preference queries, we select queries using optimal design. We now detail each component.

\subsection{Algorithm components}
\label{sec:core-components}

\paragraph{Cost estimation.}
To learn the agents' cost functions, we make the following modeling assumption.

\begin{assumption} 
    Each agent $i$ has a linear cost 
    \begin{equation}
        c_i(a):=\langle\theta_i^*,\phi_i(a)\rangle
    \end{equation}
    with $\|\theta_i^*\|\le B$. The feature maps $\phi_i:\mathcal A\!\to\!\mathbb R^d$ are continuous, the differences $\{\phi_i(a)-\phi_i(a')\mid a,a'\in\mathcal A\}$ span $\mathbb R^d$, and $\max_{a\in\mathcal A}\|\phi_i(a)\|\le L$.
	\label{as:linear-cost}
\end{assumption}

After $K$ preference queries, the planner holds for each agent $i$ the dataset
$\mathcal D_i=\{(x_{i,k},y_{i,k})\}_{k=1}^K$, where
$x_{i,k}:=\phi_i(a_{i,k})-\phi_i(a'_{i,k})$ and
$y_{i,k}=\mathds{1}(a_{i,k}\succ a'_{i,k})$. Under the Bradley--Terry model, the planner can estimate the agents' costs by minimizing the negative log-likelihood
\begin{equation}
		\hat \theta_i=\argmin_{\|\theta_i\|\leq B}\mathcal L_{\mathcal D_i}(\theta_i),
		\label{eq:mle}
	\end{equation}
	where
	\[
	\begin{aligned}
	\mathcal L_{\mathcal D_i}(\theta_i):=-\sum_{(x_{i,k}, y_{i,k})\in\mathcal D_i} &\bigl[y_{ik}\log \sigma(\langle\theta_i,x_{i,k}\rangle)\\&\hspace{-0.3cm}+(1- y_{i,k})\log(\sigma(-\langle\theta_i,x_{i,k}\rangle))\bigr].	
	\end{aligned}
	\]
\revision{While the Bradley--Terry model depends on the agents' utilities, in our algorithm preference queries are collected during exploration rounds. In these rounds, payments are either zero (one-shot game) or constant across alternatives (multi-round game), so they do not affect the reported preferences.} The resulting estimated cost function of agent $i$ is then given by
\begin{equation}
    \label{eq:esti-cost}
    \hat c_i(a; \mathcal D_i) := \langle \hat \theta_i, \phi_i(a) \rangle.
\end{equation}
Importantly, each agent's cost function is learned separately and depends only on her own feedback.

\paragraph{Optimal design.}
The quality of the estimated cost functions $\hat{c}_i$ depends on how informative the individual preference queries $a_{i, k}$ and $a'_{i, k}$ are. Since the set of all queries $\mathcal{A}\times\mathcal{A}$ could be large or infinite -- we only assumed $\mathcal{A}_i$ to be compact -- exhaustive exploration (as by \citet{kandasamy2023vcg} for numerical feedback) is not possible. Instead, we leverage optimal experimental design \citep{pukelsheim2006optimal} to select a set of queries that optimally explore the allocation space. Specifically, choosing comparison pairs via a D-optimal design subroutine \citep{lattimore2020bandit}, denoted as $\textsc{D-Optimal-Design}(K)$ in Algorithm \ref{alg:oneshot}, yields the following confidence guarantee for the estimated costs.
\begin{lemma}\label{lem:epsilon}
    Under Assumption~\ref{as:linear-cost}, supppose for each agent $i$, the planner selects $K>d(d+1)/2$ queries by $\textsc{D-Optimal-Design}(K)$. If the agents provide truthful preference feedback $\bar{\mathcal D}_i = \{(x_{i,k}, \bar y_{i, k})\}_{k=1}^K$, then with probability at least $1-\delta$,
    \begin{equation}
        \left|(c_i(a) - c_i(a')) - \left(\hat{c}_i(a; \bar{\mathcal D}_i)) -  \hat{c}_i(a'; \bar{\mathcal D}_i)\right)  \right| \leq \epsilon_K(\delta),
    \end{equation}
    with $\epsilon_K(\delta) \in \tilde{\mathcal{O}}\left( d \sqrt{\log(1/\delta) / K}\right)$.
\end{lemma}
The above result follows by combining an MLE confidence guarantee \revision{\citep{schlaginhaufen2025efficient}} with the Kiefer--Wolfowitz theorem for optimal design \citep{kiefer1960equivalence}. Remarkably, the sample complexity for estimating $c_i$ up to a given error $\epsilon$ is independent of the size of the allocation space $\mathcal{A}$ and depends only on the dimension $d$ and $\epsilon$. As shown in Appendix~\ref{app:design}, a support of $d(d+1)/2$ distinct comparisons pairs is enough, which can then be queried repeatedly. For the precise definition of the subroutine $\textsc{D-Optimal-Design}(K)$ and implementation details, see Appendix~\ref{app:design}. The full proof of Lemma~\ref{lem:epsilon} is provided in \ref{app:epsilon}.

\paragraph{VCG mechanism.}
Based on each agent's learned cost \cref{eq:esti-cost}, we compute allocations and payments using the VCG mechanism \citep{vickrey1961counterspeculation,clarke1971multipart,groves1973incentives}.
\begin{definition}[VCG mechanism]
		The allocation and payment rules are defined as
		\begin{equation}
		\label{eq:vcg-learn}
		\begin{aligned}
		\hat{a}(\mathcal D) &\in \argmin_{a\in \mathcal F}\, \hat J(a;\mathcal D) , \\
		p_{i\text{}}(\hat{a}(\mathcal D)) &= \min_{a\in \mathcal F}\, \hat J_{-i}(a; \mathcal D_{-i}) - \hat J_{-i}(\hat{a}(\mathcal D); \mathcal D_{-i}),
		\end{aligned}
		\end{equation}
		where $\mathcal D=(\mathcal D_1,\ldots,\mathcal D_N)$, $\hat J(a;\mathcal D)=\sum_{i=1}^N \hat c_i(a;\mathcal D_i)$ and $\hat J_{-i}(a;\mathcal D_{-i})=\sum_{j\neq i} \hat c_j(a;\mathcal D_j)$.
		\end{definition}
While the VCG mechanism is well understood in the classical auction setting, where agents provide bid functions, our mechanism \cref{eq:vcg-learn} operates on cost functions estimated from preference feedback. This raises the question of how the desirable properties (i.e., truthfulness, individual rationality, and efficiency) are affected by the learning errors, which we analyze next.

\subsection{Algorithm}
We now present the one-shot algorithm, summarized in
Algorithm \ref{alg:oneshot}. The central planner first chooses $K$ queries according to a D-optimal design subroutine, $\textsc{D-Optimal-Design}(K)$, which ensures sufficient exploration. The planner then collects each agent's preferences between two candidate allocations. These preferences are used to estimate the agents' cost parameter. Finally, the central planner determines VCG allocations and payments based on the agents' learned costs.

\begin{algorithm}
	\caption{One-shot game: VCG mechanism with preference feedback}
	\label{alg:oneshot}
	\begin{algorithmic}[1]
	\State \textbf{Input:} $K$
	\State Select $K$ comparisons for each agent $i$: 
        \Statex $\{(a_{i,k},a'_{i,k})\}_{k=1}^K \gets \textsc{D-Optimal-Design}(K)$.
	\State Collect preferences $\{ y_{i,k}\}_{k=1}^K$ from  each agent $i$.
	\State Estimate $\hat\theta_i$ from $\mathcal D_i$ via Equation \eqref{eq:mle}.
	\State Compute $\hat{a}(\mathcal D)$ and $p_i(\hat{a}(\mathcal D)$ for all $i$ via~\cref{eq:vcg-learn}.
	\end{algorithmic}
\end{algorithm}
\subsection{Theoretical guarantees}
\label{sec:theoretical}
In this section, we provide bounds for truthfulness, individual rationality, and efficiency in the one-shot game. These bounds reflect the worst-case outcome from the mechanism's perspective, for example, the maximum utility an agent could gain from misreporting preferences. To show individual rationality, we make the natural assumption that agents do not incur any cost when they are not part of an allocation. 
	\begin{assumption}
		There exists an allocation $a^0\in\mathcal A$ with $\phi_i(a^0)=0$, for all $i\in[N]$.
	\label{as:zero-allocation}
	\end{assumption}
 
Next, we show that Algorithm \ref{alg:oneshot} is approximately truthful, individually rational, and efficient. We use $\bar{\mathcal D}$ to denote truthful feedback according to the Bradley--Terry model \cref{eq:bradley-terry}, and $\mathcal D$ for arbitrary feedback.
\begin{theorem}
	\label{th:approximate}
		Let \cref{as:linear-cost} hold, and suppose the mechanism makes $K> d(d+1)/2$ preference queries to each agent $i\in[N]$. Then, with probability at least $1-\delta$, \cref{alg:oneshot} satisfies:
	\begin{enumerate}[label={\arabic*)}]
		\item \textbf{Truthfulness:} An agent's utility gain for arbitrary (possibly strategic) feedback is at most: $$u_i(\hat{a}(\mathcal D_i,\mathcal D_{-i}))-u_i(\hat{a}(\bar{\mathcal D_i},\mathcal D_{-i}))\leq \epsilon_K(\delta).$$
		\item \textbf{Individual rationality:} Under \cref{as:zero-allocation}, a truthful agent's utility is at least: $$u_i(\hat{a}(\bar{\mathcal D}_i,\mathcal D_{-i}))\geq-\epsilon_K(\delta).$$
		\item \textbf{Efficiency:} With all agents truthful, the gap to the optimal social cost $J(a^*)$ is at most: $$J(\hat{a}(\bar{\mathcal D}))-J(a^*)\leq N\epsilon_K(\delta/ N).$$
	\end{enumerate}
	The bound $\epsilon_K(\cdot)\in\tilde{\mathcal O}\left(d\sqrt{\log(1/\cdot)/K}\right)$ holds after $K$ queries per agent.
	\end{theorem}
    This theorem shows that truthfulness, individual rationality, and efficiency are preserved for sufficient preference queries. The error terms $\epsilon_K(\delta)$ are in line with statistical rates for learning from preferences \citep{scheid2024optimal}. Note that the error decreases as we increase the number of queries.
	
	\paragraph{Proof idea.} The full proof is given in \cref{app:approximate}. \textit{1) Truthfulness:} The argument is based on the payments of the VCG mechanism. In particular, the payment rule $p_i$ aligns agent $i$'s utility with the negative estimated social cost, that is, maximizing individual utility is equivalent to minimizing $\hat J(a;\mathcal D)$. As the mechanism selects the socially efficient allocation under truthful feedback, agents have no incentive to misreport. The only inefficiencies arise from the learning error $\epsilon_K(\delta)$, which we control using the bounds for preference-based learning from \cref{lem:epsilon}.\\
    \textit{2) Individual rationality:} The VCG payment ensures that a truthful agent has non-negative utility, deteriorated only by the learning error $\epsilon_K(\delta)$. Unlike for numerical feedback, we need \cref{as:zero-allocation} to rule out constant terms in the cost function, which would be unidentifiable from preferences.\\
    \textit{3) Efficiency:} Finally, when considering the social cost, learning errors of the $N$ agents add up. Taking the union bound over all agents results in at most a gap of $N\epsilon_K(\delta/N)$, relative to the optimal social cost.
    
	\section{MULTI-ROUND GAME: ALGORITHM AND ANALYSIS}\label{sec:multi-round}
    Next, we consider a multi-round game, which extends the one-shot game to an online setting. The mechanism proceeds in stages that consist of an \emph{exploration phase}, where the planner queries agents' preferences, followed by an \emph{exploitation phase}, where the planner implements the allocation according to the VCG mechanism. This setting is closely related to no-regret learning in linear bandits \citep{lattimore2020bandit}, where the learner must carefully balance exploration and exploitation to select informative queries while ensuring low regret. An example of such a setting is local electricity markets, where the planner must learn consumers' preferences for energy adjustments online, while also deploying such adjustments in real-time to ensure grid stability.

    In order to study the multi-round game, we extend the results from the one-shot game. As before, we assume a linear cost model (Assumption \ref{as:linear-cost}) for agents, \revision{with the cost parameters fixed across rounds}.

    \subsection{Algorithm}
    We introduce the multi-round algorithm, summarized in \cref{alg:multiround}. Inspired by \citet{kandasamy2023vcg}, the algorithm alternates between exploration phases for learning and progressively longer exploitation phases with VCG allocations and payments. In contrast to their work, which assumes that agents give numerical feedback over finitely many allocations, we consider preference feedback over a compact allocation space. Their analysis does therefore not extend trivially to our setting and requires the methods developed in \cref{sec:one-shot}.
    
    The multi-round algorithm differs from the one-shot setting in several ways. During exploration, the planner makes payments to the agents to ensure individual rationality in every round. These payments are set to $2c_{\max} = 2BL$, the maximum cost under \cref{as:linear-cost}, for each queried allocation pair. Moreover, each stage $s$ consists of $K$ exploration rounds and $M_s$ exploitation rounds, where $M_s = \lfloor \frac{5}{6} K \sqrt{s} \rfloor$, so that exploitation phases progressively become longer as estimates improve. 
    
	\begin{algorithm}
	\caption{Multi-round game: VCG mechanism with preference feedback}
	\label{alg:multiround}
	\begin{algorithmic}[1]
    \State \textbf{Input:} $K$, $T$
    \State $t \leftarrow 1$, $s \leftarrow 1$
    \While{$t \leq T$}
    \Statex \textcolor{blue}{\phantom{whi}\texttt{/* Exploration phase\hfill/*}}
    \State \parbox[t]{0.85\columnwidth}{Select $K$ comparisons for each agent $i$: 
        \Statex $\{(a_{i,k},a'_{i,k})\}_{k=1}^K \gets \textsc{D-Optimal-Design}(K)$.}
    \State  \parbox[t]{0.85\columnwidth}{Collect queried preferences $\{y_{i,k}^s\}_{k=1}^K$ from each agent $i$.}
    \State  Pay $2Kc_{\max}$ to each agent $i$.
    \State $t \leftarrow t + K$
    \State \parbox[t]{0.85\columnwidth}{Estimate $\hat\theta_i^s$ from $\mathcal D_i^s$ via Equation \cref{eq:mle}.}
    \Statex \textcolor{blue}{\phantom{whi}\texttt{/* Exploitation phase\hfill /*}}
    \State \parbox[t]{0.85\columnwidth}{Compute $\hat{a}_t(\mathcal D^s)$ and $p_i(\hat a_t (\mathcal D^s))$ for all $i$ via \cref{eq:vcg-learn} during $M_s=\lfloor\frac{5}{6} K\sqrt{s}\rfloor$ rounds.}
    \State $t \leftarrow t + M_s$
    \State $s \leftarrow s + 1$
    \EndWhile
	\end{algorithmic}
	\end{algorithm}
    
    \subsection{Theoretical guarantees}
    In the multi-round game, we focus on cumulative quantities over rounds $t = 1, \dots, T$. Let the cumulative utility of agent $i$ be defined as $\bar U_{i}(T)=\sum_{t=1}^T u_i(\hat a_t(\bar{\mathcal D_i},\mathcal D_{-i}))$, where agent $i$ is truthful in every round. Similarly, let $U_{i}(T)= \sum_{t=1}^T u_i(\hat a_t(\mathcal D_i,\mathcal D_{-i}))$ be the cumulative utility of agent $i$, when deviating from truthful feedback in at least one round. Furthermore, we define welfare regret as:
	\begin{equation}
	\label{eq:welfare-regret}
		R^w(T)=\sum_{t=1}^T \bigl(J(\hat a_t(\bar{\mathcal D}))-J(a^*)\bigr),
	\end{equation}
	where $\hat a_t(\bar{\mathcal D})$ is the mechanism's allocation computed from truthful feedback and $a^*$ is the optimal allocation. Note that $R^w(T)$ also includes allocation costs of agents during exploration rounds.

    Our goal is to provide regret bounds for these cumulative quantities, again from the mechanism's perspective. We show sublinear regret bounds, which imply asymptotic guarantees for truthfulness, individual rationality, and efficiency.
    
	\begin{theorem} \label{th:asymptotic} Let \cref{as:linear-cost} hold, and suppose the mechanism makes $K> d(d+1)/2$ preference queries to each agent $i\in[N]$. Then, after $T$ rounds with probability $1-\delta$, \cref{alg:multiround} satisfies:
	\begin{enumerate}[label={\arabic*)}]
		\item \textbf{Truthfulness:} An agent's utility gain from deviating in at least one round is upper bounded by: $$U_{i}(T)-\bar U_{i}(T)\in \tilde{\mathcal O}\left(dT^{2/3}\sqrt{\log(1/\delta)}\right).$$
		\item \textbf{Individual rationality:} Under \cref{as:zero-allocation}, a truthful agent's utility is lower bounded by: $$-\bar U_{i}(T)\in \tilde{\mathcal O}\left(dT^{2/3}\sqrt{\log(1/\delta)}\right).$$
		\item \textbf{Efficiency:} With all agents truthful, the welfare regret \cref{eq:welfare-regret} is upper bounded by: $$R^w(T)\in \tilde{\mathcal O}\left(NdT^{2/3}\sqrt{\log(N/\delta)}\right).$$
	\end{enumerate}
	\end{theorem}
    Compared to the approximate guarantees in the one-shot game, this theorem implies asymptotic truthfulness, individual rationality and efficiency. This follows directly from sublinearity of the upper bounds, for example, $R^w(T)/T\rightarrow 0$ as $T\rightarrow\infty$ for welfare regret. These rates are in line with works that consider numerical feedback and show sublinear upper bounds of $\tilde O(T^{2/3})$ \citep{babaioff2010truthful}. This aligns with the broader insight that, under the Bradley--Terry model, learning from preference feedback is not fundamentally harder than learning from numerical feedback \citep{ailon2014reducing}.
	\paragraph{Proof idea.} The full proof is given in \cref{app:asymptotic}. \textit{1) Truthfulness:} Utilities in exploration rounds are not affected by strategic behavior as allocations and payments are selected independently of the preference feedback. In exploitation rounds, the utility gain per round is upper bounded by \cref{th:approximate} with $\epsilon_{Q_t}(\delta)$, where $Q_t$ denotes the number of preference queries at time $t$. Summing over all rounds $T$ yields the finite-time bound.\\ \textit{2) Individual rationality:} Fixed payments of $2c_{\max}$ for a pairwise comparison guarantee non-negative utility in exploration rounds. In exploitation rounds, learning errors of $\epsilon_{Q_t}(\delta)$ per round sum up as above.\\ \textit{3) Efficiency:} Due to repeated exploration, the learning error decreases over time, thus the mechanism converges to efficient allocations. The convergence rate is given by $M_s$, which ensures sublinear welfare regret by progressively longer exploitation phases as estimates get better.

    Our theoretical results show that even when agents only provide preference feedback, a mechanism can be implemented that converges to a socially efficient allocation. A practical application for this setting arises in local electricity markets, which we discuss next.
     
\section{EXPERIMENTS}
\label{sec:experiments}
An emerging challenge with the increased penetration of intermittent renewable energy in electric grids is the balancing of energy supply and demand. A promising approach is demand response, where a grid operator procures energy flexibility from domestic consumers through recurring \textit{demand response events}. The allocation of this flexibility can be organized as local electricity markets \citep{tsaousoglou2022market}. This naturally raises questions about strategic behavior among consumers in these markets. Accordingly, a range of works study this problem and propose various market mechanisms \citep{tsaousoglou2021mechanism,fochesato2022stackelberg,crowley2025can}. However, an open challenge is how consumers can communicate the costs they incur when providing energy flexibility, which are difficult to quantify \citep{abedrabboh2023applications}. We address this challenge in our setting, where consumers express preferences over allocations of energy flexibility.
\subsection{Simulation setup}
We consider a group of consumers who adjust their thermal loads, such as heating and cooling, during demand response events. The consumers participate in a local electricity market, where a central planner coordinates the demand response. The planner allocates the energy flexibility $a=(a_1,\dots,a_N)$ with the goal to minimize the social cost. The individual cost of consumers is given by their discomfort when the room temperature deviates from their preferred temperature due to their energy flexibility $a_i$ (kWh). We assume that each consumer has up to $d$ rooms, which each contributes to the discomfort as follows:
\[
	c_i(a) = \langle \theta_i^*, \phi_i(a_i) \rangle=\sum_{l=1}^d\theta_{i,l}^* a_{i,l}^2,
\]
where $\theta_i^*$ (\$/kWh$^2$) are positive cost parameters and $\phi_i(a_i)$ (kWh$^2$) is a quadratic feature map \citep{li2011optimal}.\footnote{While our theoretical guarantees allow costs $c_i$ to depend on the entire allocation $a$, our simulations consider restricted feature maps $\phi_i(a_i)$ for illustration purposes.} When a demand response event occurs, the central planner must choose an allocation from the feasible set $\mathcal F$, defined as
\[\mathcal F=\left\{a\in\mathcal A\;\middle|\;\sum_{i=1}^N\sum_{l=1}^d a_{i,l}=P\right\},\]
where $P$ (kWh) is the total energy flexibility required by the grid. In this setting, preference queries can be thought of as simulated demand response events, such as questionnaires or tests.
\subsection{Numerical results}
\label{sec:numerical}
In the simulation, energy flexibilities belong to a discrete allocation set of size $|\mathcal A_i|=15$ for each agent $i$. \revision{Note that this implies $15^2$ possible comparisons per agent, making exhaustive exploration infeasible. The discretization allows us to compute an exact D-optimal design, even though the theoretical results hold for any compact allocation space. See \cref{app:design} for a discussion of approximate optimal designs in the case of infinite sets.} The cost parameters are generated synthetically with random samples between 0.1 and 0.5\;\$/kWh$^2$ \citep{safdar2019costs}, reflecting heterogeneity across consumers and rooms. We assume a constant flexibility requirement $P=15$\;kWh for each demand response event. An overview of parameters and runtimes is given in \cref{app:experiments}. Next, we present our simulation results, illustrating truthfulness and efficiency.
 
To assess truthfulness in the one-shot game, we compare a given agent’s utility under truthful and untruthful preference feedback. We focus on untruthful feedback where preferences are sampled from the Bradley--Terry model \cref{eq:bradley-terry}, but for biased cost parameters $\theta_i=\theta_i^*+\Delta\theta$, where $\Delta\theta\in\{-0.2,-0.1,0.1,0.2\}$. As a benchmark, we compare utilities for allocations and payments computed under the pay-as-bid (first-price) mechanism based on the learned costs. Figure~\ref{fig:truthfulness} shows that, under the pay-as-bid mechanism, agents benefit from overstating their costs. Under the VCG mechanism, however, untruthful feedback can only increase an agent's utility when the number of queries $K$ is small.

\begin{figure}[t!]
	\centering
	\vspace{.3in}
    \includegraphics[width=\columnwidth]{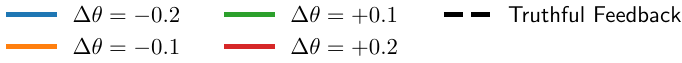}

    \vspace{0.1cm}
    
    \begin{subfigure}[t]{0.5\columnwidth}
    \includegraphics[width=1.03\textwidth]{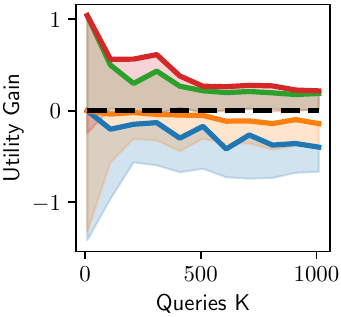}
    \caption{Pay-as-Bid}
    \label{fig:utility-vcg}
\end{subfigure}
\begin{subfigure}[t]{0.48\columnwidth}
    \includegraphics[width=1.04\textwidth]{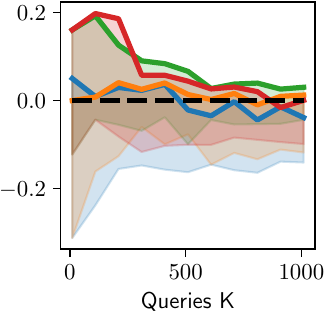}
    \caption{VCG}
    \label{fig:utility-pay-as-bid}
\end{subfigure}

	\vspace{.3in}
	\caption{Impact of misreporting on a given agent's utility in the one-shot game under (a) the pay-as-bid mechanism and (b) the VCG mechanism. We show the utility gain $u_i(\hat a(\mathcal D_i,\mathcal D_{-i})) - u_i(\hat a(\bar{\mathcal D}_i,\mathcal D_{-i}))$ across independent runs, with the worst-case outcome highlighted.}
	\label{fig:truthfulness}
\end{figure}

\begin{table}
\centering
\caption{Maximum utility gain of a given agent with untruthful feedback for $K=1000$ queries.}
\begin{tabular}{c|cc|cc}
\textbf{Deviation} & \multicolumn{2}{c|}{\textbf{Pay-as-Bid}} & \multicolumn{2}{c}{\textbf{VCG}} \\

 & Max & Std & Max & Std \\
\hline\hline
$-0.2$ & $-0.40$ & 0.07 & $-0.04$ & 0.03 \\
$-0.1$ & $-0.14$ & 0.06 & $+0.01$ & 0.03 \\
$+0.1$ & $+0.19$ & 0.05 & $+0.03$ & 0.02 \\
$+0.2$ & $+0.22$ & 0.05 & $+0.01$ & 0.03 \\
\end{tabular}
\label{tab:gain}
\end{table}

\cref{tab:gain} provides an overview of a given agent's maximum utility gain for $K=1000$ preference queries. Under the VCG mechanism, the incentive to deviate is negligible, once sufficient preferences are queried. Standard deviations are small, indicating that the agent's utility gains across runs become concentrated as the number of queries increases. This trend generalizes to any agent, confirming that the mechanism elicits truthful preference feedback.

Finally, we assess efficiency under truthful feedback both in the one-shot and multi-round settings. In the one-shot game, Figure \ref{fig:efficiency-one-shot} shows that the worst-case social cost gap decreases at a rate of $\tilde{\mathcal O}(K^{-1/2})$. In the multi-round game, Figure \ref{fig:efficiency-multi-round} shows that the average welfare regret decreases a rate of $R^w(T)/T\in\tilde{\mathcal O}(T^{-1/3})$. Both rates are consistent with our theoretical guarantees. Overall, these results confirm that the mechanism converges to a socially efficient allocation by learning agents' costs from preference feedback.

\begin{figure}
	\centering
	\vspace{.3in}
    \begin{subfigure}[t]{0.48\columnwidth}
    \includegraphics[width=\textwidth]{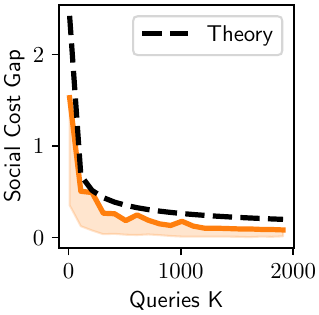}
    \caption{One-Shot Game}
    \label{fig:efficiency-one-shot}
\end{subfigure}
\hfill
\begin{subfigure}[t]{0.5\columnwidth}
    \includegraphics[width=\textwidth]{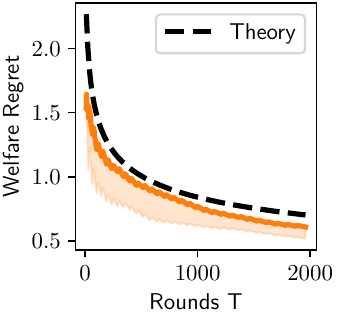}
    \caption{Multi-Round Game}
    \label{fig:efficiency-multi-round}
\end{subfigure}
	
	\vspace{.3in}
	\caption{Mechanism's performance in terms of (a) social cost gap $J(\hat a(\bar{\mathcal D}))-J(a^*)$ in the one-shot game and (b) average welfare regret $R^w(T)/T$ in the multi-round game. Both are normalized by the optimal social cost $J(a^*)$.}
	\label{fig:efficiency}
\end{figure}

\section{CONCLUSION}
In this paper, we studied the problem of designing preference queries and payment rules to ensure truthfulness, individual rationality, and efficiency. We proposed an algorithm that integrates preference-based learning with VCG, and proved that it satisfies these properties approximately in a one-shot setting and asymptotically in an online setting. Our method improves practicality by removing the need for bid functions or numerical feedback, while accommodating large allocation spaces. Lastly, we validated our method in a numerical case study on demand response in local electricity markets, demonstrating its ability to elicit truthful preferences and converge to socially efficient allocations.

\revision{Based on these results, several directions for future work emerge. First, our mechanism builds on VCG, which is known to suffer from limitations such as lack of budget balance and potential revenue loss for the central planner \citep{ausubel2006lovely}. Future work could therefore explore alternative mechanism designs, such as core-selecting or ascending auctions \citep{ausubel2002ascending}, within a preference-based learning setup. Second, while the Bradley--Terry model is a standard cardinal preference model widely used in the literature, it remains restrictive in settings where preferences exhibit richer structure. An interesting direction would be to incorporate models that account for full reward distributions or alignment distortions, as recently studied by \citet{golz2025distortion}. Extending the theoretical guarantees to more general models, including time-varying costs or richer function classes (e.g., RKHS), and establishing robustness guarantees under model misspecification remain important open challenges. Finally, validating the approach using real-world preference data in simulations  can highlight its practical effectiveness.}

\subsubsection*{Acknowledgments}
The authors thank the reviewers for their constructive and insightful feedback. The authors also thank the members of the Sycamore lab for valuable discussions and feedback. This work was supported in part by the grant Learning for Safety in Data-Driven Control (Grant No. 200020\_207984).

\bibliography{bibliography.bib}

\section*{Checklist}
 \begin{enumerate}

 \item For all models and algorithms presented, check if you include:
 \begin{enumerate}
   \item A clear description of the mathematical setting, assumptions, algorithm, and/or model. [Yes]
   \item An analysis of the properties and complexity (time, space, sample size) of any algorithm. [Yes]
   \item (Optional) Anonymized source code, with specification of all dependencies, including external libraries. [Yes]
 \end{enumerate}

 \item For any theoretical claim, check if you include:
 \begin{enumerate}
   \item Statements of the full set of assumptions of all theoretical results. [Yes]
   \item Complete proofs of all theoretical results. [Yes]
   \item Clear explanations of any assumptions. [Yes]   
 \end{enumerate}

 \item For all figures and tables that present empirical results, check if you include:
 \begin{enumerate}
   \item The code, data, and instructions needed to reproduce the main experimental results (either in the supplemental material or as a URL). [Yes]
   \item All the training details (e.g., data splits, hyperparameters, how they were chosen). [Not Applicable]
         \item A clear definition of the specific measure or statistics and error bars (e.g., with respect to the random seed after running experiments multiple times). [Yes]
         \item A description of the computing infrastructure used. (e.g., type of GPUs, internal cluster, or cloud provider). [Yes]
 \end{enumerate}

 \item If you are using existing assets (e.g., code, data, models) or curating/releasing new assets, check if you include:
 \begin{enumerate}
   \item Citations of the creator If your work uses existing assets. [Not Applicable]
   \item The license information of the assets, if applicable. [Not Applicable]
   \item New assets either in the supplemental material or as a URL, if applicable. [Not Applicable]
   \item Information about consent from data providers/curators. [Not Applicable]
   \item Discussion of sensible content if applicable, e.g., personally identifiable information or offensive content. [Not Applicable]
 \end{enumerate}

 \item If you used crowdsourcing or conducted research with human subjects, check if you include:
 \begin{enumerate}
   \item The full text of instructions given to participants and screenshots. [Not Applicable]
   \item Descriptions of potential participant risks, with links to Institutional Review Board (IRB) approvals if applicable. [Not Applicable]
   \item The estimated hourly wage paid to participants and the total amount spent on participant compensation. [Not Applicable]
 \end{enumerate}

 \end{enumerate}
\onecolumn
\aistatstitle{Supplementary Material}
\appendix
\crefalias{section}{appendix}
\crefalias{subsection}{appendix}

\section{Technical results}
This section presents the technical results necessary for our theorems’ proofs.
\subsection{Confidence set}
    First, we state an established lemma for maximum likelihood estimation given preference feedback according to the Bradley--Terry model \cref{eq:bradley-terry}. \revision{We include it for completeness, as it is used to control the learning error in our theorems.}
	\begin{lemma}[\citet{schlaginhaufen2025efficient}, Lemma A.1]
	\label{lem:schlaginhaufen}
	Let \cref{as:linear-cost} hold and let $\hat{\theta}_i$ be defined as in Equation \eqref{eq:mle}. Then, after $K$ queries, we get with probability at least $1-\delta$:
	\begin{equation}
	\|\hat\theta_i-\theta_i^*\|_{V_{i,K}}\leq\gamma_{i,K}(\delta).
	\end{equation}
	The bound is given as
	\[
	\begin{aligned}
	\gamma_{i,K}^2(\delta)=\kappa_i\biggl[\log\left(\frac{1}{\delta}\right)+d\log\left(\max\left\{e,\frac{4eBL(K-1)}{d}\right\}\right)\biggr].	
	\end{aligned}
	\]
    where $\kappa_i=\max_{\|\theta_i\|\leq B,\,x_i\in\mathcal X_i} \frac{1}{\dot \sigma(\langle\theta_i,x_i\rangle)}$.
	\end{lemma}
	Here, $\|x\|_A := \langle x,Ax\rangle$ denotes the Mahalanobis norm given a positive definite matrix $A \in \mathbb R^{d\times d}$. This norm with respect to the empirical design matrix $V_{i,K}:=\sum_{k=1}^K x_{i,k}\,x_{i,k}^\top$ measures the learning error for a given query $x_{i,k}:=\phi_i(a_{i,k})-\phi_i(a_{i,k}')$. The leading factor $\kappa_i$ characterizes the difficulty of learning from preferences when the feedback is nearly deterministic. 
\subsection{Optimal experimental design}
\label{app:design}
Since the set of all queries  $\mathcal X_i =\{\phi_i(a') - \phi_i(a):a,a'\in\mathcal A\}$ could be large or infinite, we leverage optimal experimental design to select a set of queries that optimally explore the allocation space $\mathcal A$. \revision{We want to control the worst-case prediction variance
\[
\max_{x_i \in \mathcal X_i} \|x_i\|^2_{V_i(\pi_i)^{-1}}, \quad \text{with} \quad V_i(\pi_i):= \int_{\mathcal X_i} x_i x_i^\top d\pi_i(x_i),
\]
where $V_i(\pi_i)$ is referred to as the theoretical design matrix associated with a distribution of queries $\pi_i$. While directly optimizing this G-optimality criterion is in general not straightforward, the Kiefer--Wolfowitz theorem \citep{kiefer1960equivalence} states the equivalence between G-optimality and D-optimality.} We thus compute a D-optimal design
\begin{equation}
	\pi_i^*\in\argmax_{\pi_i\in\Delta_{\mathcal X_i}}\log \det V_i(\pi_i),
	\label{eq:d-optimal}
\end{equation}
where $\Delta_{\mathcal X_i}$ is the set of all probability measures supported on the compact subset $\mathcal X_i\subset \mathbb R^d$. Note that compactness of the design space $\mathcal X_i$ follows from continuity of $\phi_i$ and compactness of $\mathcal A$. We can solve \cref{eq:d-optimal} efficiently for finite allocation sets $\mathcal A$, e.g., via the Frank-Wolfe algorithm \citep{frank1956algorithm}. \revision{For infinite sets, Corollary 4.1 of \citep{hazan2016volumetric} shows that, given access to an argmax oracle that can identify the optimal allocation, one can find a $\mathcal O(\sqrt{d})$-approximate G-optimal design of support $d$ using $\mathcal O (d^2\log d)$ calls to the argmax oracle.} \cref{alg:optimal-design} details the subroutine $\textsc{D-Optimal-Design}(K)$ to select the queries.
\begin{algorithm}
	\caption{$\textsc{D-Optimal-Design}(K)$}
	\label{alg:optimal-design}
	\begin{algorithmic}[1]
    \State \textbf{Input:} $K$
    \State Compute optimal design $\pi_i^*$ \cref{eq:d-optimal} with $|\text{Supp}(\pi_i^*)| \le d(d+1)/2$.
    \State Set $n\gets K-d(d+1)/2$.
    \State Round number of queries $n_{x_{i}} \gets \lceil n \pi_i^*(x_i) \rceil$ for comparison $x_i$.
    \State Select queries $\{(a_{i,k},a'_{i,k})\}_{k=1}^K$ from the rounded design.
    \end{algorithmic}
\end{algorithm}

Our lemma below bounds the prediction variance when queries are selected based on a D-optimal design. The result is shown given an exact $\pi_i^*$, but can be extended for approximations. This lemma is used to control the learning error in our theorems, together with \cref{lem:schlaginhaufen}.

\begin{lemma}
	\label{lem:design}
	Suppose $K> d(d+1)/2$ queries are selected using \cref{alg:optimal-design}. Then, the worst-case prediction variance is bounded as:
	\begin{equation}
		\max_{x_i\in\mathcal X_i}\|x_i\|^2_{V_{i,K}^{-1}}\leq \frac{d}{K-d(d+1)/2}.
	\end{equation}
	\end{lemma}
\begin{proof} Our goal is to bound the worst-case prediction variance
	\[\max_{x_i\in\mathcal X_i}\|x_i\|^2_{V_{i,K}^{-1}},\]
where $V_{i,K}=\sum_{k=1}^K x_{i,k}\,x_{i,k}^\top$ is the \textit{empirical} design matrix. Our argument builds on the Kiefer--Wolfowitz theorem \citep{kiefer1960equivalence}, which characterizes the properties of D-optimal designs with respect to the \textit{theoretical} design matrix $V_i(\pi_i)$, and follows the ideas in \citet[Section 21.1]{lattimore2020bandit} to relate the theoretical and empirical design matrices. By the Kiefer--Wolfowitz theorem, a D-optimal design $\pi_i^*$ satisfies
\begin{equation}
    \max_{x_i\in\mathcal X_i} \|x\|^2_{V_i(\pi_i^*)^{-1}} = d,
    \label{eq:kiefer-d}
\end{equation}
where $V_i(\pi_i)= \int_{\mathcal X_i} x_i x_i^\top d\pi_i(x_i)$ is the theoretical design matrix, based on the distribution $\pi_i$ rather than sampled queries. Moreover, the theorem states that there exists a D-optimal design $\pi_i^*$ with support
\begin{equation}
    |\text{Supp}(\pi_i^*)| \le d(d+1)/2.
    \label{eq:kiefer-support}
\end{equation}
	This result bounds the worst-case prediction variance for the theoretical design matrix, which we now relate to the empirical design matrix. Let $n_{x_i}=\lceil n\pi_i^*(x_i)\rceil$ denote the number of times we query $x_i$. Then, the empirical design matrix is
	\[V_{i,K}=\sum_{x_i\in\text{Supp}(\pi_i^*)} n_{x_i} x_ix_i^\top= \sum_{x_i\in\text{Supp}(\pi_i^*)}\lceil n\pi_i^*(x)\rceil x_ix_i^\top\succeq n V_i({\pi_i^*}) = n \sum_{x_i\in\text{Supp}(\pi_i^*)} \pi_i^*(x_i) x_i x_i^\top.\]
	By \eqref{eq:kiefer-d}, it follows that
	\[
		\max_{x_i\in\mathcal X_i}\|x_i\|^2_{V_{i,K}^{-1}}\leq\frac{1}{n}\max_{x_i\in\mathcal X_i}\|x_i\|^2_{V_i(\pi_i^*)^{-1}}=\frac{d}{n}.
	\]
	The rounded design results in a gap between $n$ and $K$, which we express using \eqref{eq:kiefer-support} as
	\[K=\sum_{x_i\in\text{Supp}(\pi_i^*)}n_{x_i}\leq\sum_{x_i\in\text{Supp}(\pi_i^*)}(n\pi_i^*(x)+1)\leq n+\frac{d(d+1)}{2},\]
	so that $n \geq K - d(d+1)/2$. Substituting into the previous bound yields the lemma.   
\end{proof}
\subsection{Proof of \cref{lem:epsilon}}
\label{app:epsilon}
\begin{proof}
We want to bound the learning error for pairwise preferences:
\begin{align}
\left|(c_i(a) - c_i(a')) - \left(\hat{c}_i(a; \bar{\mathcal D}_i)) -  \hat{c}_i(a'; \bar{\mathcal D}_i)\right)  \right|&=\left|\left(\hat{c}_i(a'; \bar{\mathcal D}_i) - c_i(a')\right) +\left(\hat{c}_i(a; \bar{\mathcal D}_i) - c_i(a)\right)  \right|\\&=\left|\langle\hat\theta_i-\theta_i^*,\phi_i(a')\rangle-\langle\hat\theta_i-\theta_i^*,\phi_i(a)\rangle\right|\\&= \left|\langle\hat\theta_i-\theta_i^*,\phi_i(a')-\phi_i(a)\rangle\right|,\label{eq:linear-difference}
\end{align}
which we rewrote using \cref{as:linear-cost}. Note that $\hat\theta_i$ is the cost parameter learned from truthful feedback. Then, by the Cauchy-Schwarz inequality, we have:
\[
\begin{split}
\left|(c_i(a) - c_i(a')) - \left(\hat{c}_i(a; \bar{\mathcal D}_i)) -  \hat{c}_i(a'; \bar{\mathcal D}_i)\right)  \right|&\leq \|\hat\theta_i-\theta_i^*\|_{V_{i,K}}\|\phi_i(a')-\phi_i(a)\|_{V_{i,K}^{-1}}\\&\leq\epsilon_K(\delta),
\end{split}
\]
where
\begin{equation}
\label{eq:epsilon}
\epsilon_K(\delta):=\gamma _{K}(\delta)\sqrt{\frac{d}{K-d(d+1)/2}}.
\end{equation}
Here, we combined the confidence bound for $\|\hat\theta_i-\theta_i^*\|_{V_{i,K}}$ (\cref{lem:schlaginhaufen}) and the prediction variance bound for $\|\phi_i(a')-\phi_i(a)\|_{V_{i,K}^{-1}}$ (\cref{lem:design}). We set $\gamma_K(\delta)=\max_i\gamma_{i,K}(\delta)$, with
	
\begin{equation}
\gamma_{i,K}(\delta)=\sqrt{\kappa_i\left[\log\left(\frac{1}{\delta}\right)+d\log\left(\max\left\{e,\frac{4eBL(K-1)}{d}\right\}\right)\right]}.
\label{eq:beta}
\end{equation}
	Hence, with probability at least $1-\delta$, after $K$ comparisons we have
	\[\epsilon_K(\delta)\in\mathcal O \left(\sqrt{\left[\frac{d\log(1/\delta)+d^2\log(K)}{K}\right]}\right).\]
\end{proof}
\section{Proof of \cref{th:approximate}}
\label{app:approximate}
\begin{proof}
    We prove approximate truthfulness, individual rationality, and efficiency in the one-shot game. For each property, the VCG mechanism reduces the analysis to bounding the learning error, which we control via \cref{lem:epsilon}.
		\paragraph{Truthfulness.}
		We aim to bound agent $i$'s utility gain from giving untruthful feedback, given other agents' arbitrary feedback $\D_{-i}$. On the one hand, if agent $i$ provides truthful feedback $\bar{\D}_i$, the allocation $\hat a:=\hat a(\bar \D_i,\D_{-i})$ minimizes
        \begin{equation}
            \label{eq:truthful-minimizer}
          \hat a\in\argmin_{a\in\mathcal F}\left(\hat J_{-i}(a;\D_{-i})+\hat c_i(a;\bar{\D}_i)\right).  
        \end{equation}
        On the other hand, if agent $i$ gives arbitrary feedback $\D_i$, the allocation $\hat a':=\hat a(\D_i,\D_{-i})$ minimizes
        \[\hat a'\in\argmin_{a\in\mathcal F}\left(\hat J_{-i}(a;\D_{-i})+\hat c_i(a;\D_i)\right).\]
        Now, consider agent $i$'s utility difference under arbitrary and truthful feedback:
        \[
	\begin{split}
			u_i(\hat a')-u_i(\hat a)&\overset{(i)}{=}\left[\min_{a\in\mathcal A}\hat J_{-i}(a;\D_{-i})-\hat J_{-i}(\hat a';\D _{-i})-c_i(\hat a')\right]\\&\qquad-\left[\min_{a\in\mathcal A}\hat J_{-i}(a;\D_{-i})-\hat J_{-i}(\hat a;\D_{-i})-c_i(\hat a)\right]\\&\overset{(ii)}{=}-\hat J_{-i}(\hat a';\D_{-i})-c_i(\hat a')+\hat J_{-i}(\hat a;\D_{-i})+c_i(\hat a)\\&\overset{(iii)}{=}\left[-\hat J_{-i}(\hat{a}';{\D}_{-i})-\hat c_i(\hat{a}';\bar{\D}_i)+\hat c_i(\hat{a}';\bar{\D}_i)-c_i(\hat{a}')\right]\\&\qquad+\left[\hat J_{-i}(\hat{a};{\D}_{-i})+\hat c_i(\hat{a};\bar{\D}_i)-\hat c_i(\hat{a};\bar{\D}_i)+c_i(\hat{a})\right].
	\end{split}
	\]
	 In (i), we applied the VCG payment rule \cref{eq:vcg-learn}. The term $\min_{a\in\mathcal F}\hat J_{-i}(a;\mathcal D_{-i})$ canceled in (ii), as it is independent from agent $i$'s feedback. Then, (iii) added and subtracted the terms $\hat c_i(\hat a';\bar{\D}_i)$ and $\hat c_i(\hat a;\bar{\mathcal D}_i)$, the agent's cost estimated under truthful feedback, but for allocations $\hat a'$ (untruthful feedback) and $\hat a$ (truthful feedback). Since $\hat a$ is the truthful minimizer in \cref{eq:truthful-minimizer}, we have $\hat J_{-i}(\hat{a};{\D}_{-i})+\hat c_i(\hat{a};\bar{\D}_i)\leq\hat J_{-i}(\hat{a}';{\D}_{-i})+\hat c_i(\hat{a}';\bar{\D}_i)$, and hence:
    \[
    \begin{split}
		u_i(\hat{a}')-u_i(\hat{a})&\leq \hat c_i(\hat{a}';\bar\D_i)-c_i(\hat{a}')-\hat c_i(\hat{a};\bar \D_i)+c_i(\hat{a})\\&\leq \left|\left(c_i(\hat{a})-c_i(\hat{a}')\right)-\left(\hat c_i(\hat{a};\bar \D_i)-\hat c_i(\hat{a}';\bar\D_i)\right)\right|.
    \end{split}
	\]
     But this is exactly the term we bound in \cref{lem:epsilon}, so with probability at least $1-\delta$,
     \begin{equation}
     \label{eq:oneshot-truthful}
    u_i(\hat{a}')-u_i(\hat{a})\leq\epsilon_K(\delta).
     \end{equation}
     \[\]
	
	\paragraph{Individual rationality.} Next, we show that a truthful agent's utility is lower bounded. Letting $\hat a:=\hat a(\bar\D_i,\D_{-i})$, we have
	\[
	\begin{split}
		u_i(\hat{a})&\overset{(i)}{=}\left(\min_{a\in\mathcal F}\hat J_{-i}(a;\mathcal D_{-i})-\hat J_{-i}(\hat{a};\mathcal D_{-i})\right)-c_i(\hat{a})\\&
		\overset{(ii)}{=} \min_{a\in\mathcal F}\hat J_{-i}(a;\mathcal D_{-i})-\hat J_{-i}(\hat{a};\mathcal D_{-i})-\hat c_i(\hat{a};\bar\D_i)+\hat c_i(\hat{a};\bar\D_i)-c_i(\hat{a}),
	\end{split}
	\]
	where in (i) we substituted the VCG payments in the utility and in (ii) $\hat c_i(\hat{a};\bar\D_i)$ was added and subtracted. Observe that $\min_{a\in\mathcal F}\left(\hat J_{-i}(\hat{a};\D_{-i})+\hat c_i(a;\bar \D_i)\right)\leq\min_{a\in\mathcal F}\hat J_{-i}(a;\mathcal D_{-i})$, because the mechanism can always ignore agent $i$ if including them does not reduce total cost. Their participation can only lower or maintain the social cost, thus
	\[
	\begin{split}
		u_i(\hat{a})&
		\geq \hat c_i(\hat{a};\bar\D_i)-c_i(\hat{a})\\&
		=\hat c_i(\hat{a};\bar\D_i)-c_i(\hat{a})+c_i(a^0)-\hat c_i(a^0;\bar \D_i)\\&\geq-\left|\left(c_i(a^0)-c_i(\hat{a})\right)-\left(\hat c_i(a^0;\bar \D_i)-\hat c_i(\hat{a};\bar\D_i)\right)\right|,
	\end{split}
	\]
    where we added $c_i(a^0)-\hat c_i(a^0;\bar \D_i)$, both zero under \cref{as:zero-allocation}. Then, by \cref{lem:epsilon}, we get with probability at least $1-\delta$:
    \begin{equation}
		u_i(\hat{a})
		\geq -\epsilon_K(\delta).
        \label{eq:oneshot-ir}
    \end{equation}
	
	\paragraph{Efficiency.}
	Finally, we upper bound the gap to the optimal social cost $J(a^*)$ under truthful feedback. Defining $\hat a :=\hat a (\bar \D)$, we have:
	\[
	\begin{split}
		J(\hat{a})-J(a^*)&\overset{(i)}{=}J(\hat{a})-\hat J(a^*;\bar\D)+\hat J(a^*;\bar\D)-J(a^*)\\&
		\overset{(ii)}{\leq} J(\hat{a})-\hat J(\hat{a};\bar\D)+\hat J(a^*;\bar\D)-J(a^*).
		\end{split}
		\]
		
	In (i), we added and subtracted $\hat J(a^*;\bar\D)$. Then, (ii) used that $\hat J(\hat{a};\mathcal D)\leq \hat J(a^*;\mathcal D)$, because $\hat{a}$ is the minimizer of $J(\hat{a};\bar\D)$. Separating the sum into individual costs gives us
    \[
    \begin{split}
    J(\hat{a})-J(a^*)&\leq\sum_{i=1}^N \left(c_i(\hat{a})-\hat c_i(\hat{a};\bar\D_i)+ \hat c_i(a^*;\bar\D_i)-c_i(a^*)\right)\\&\leq \sum_{i=1}^N \left|\left(c_i(\hat{a})-c_i(a^*)\right)-\left(\hat c_i(\hat{a};\bar\D_i)- \hat c_i(a^*;\bar\D_i)\right)\right|.
    \end{split}
    \]
	By \cref{lem:epsilon} and taking the union bound over $N$ agents, we obtain with probability at least $1-\delta$,
	\begin{equation}
		J(\hat{a})-J(a^*)\leq N\epsilon_K(\delta/N).
		\label{eq:oneshot-efficient}
	\end{equation}
    \end{proof}
\section{Proof of \cref{th:asymptotic}}
\label{app:asymptotic}
\begin{proof}
We prove asymptotic truthfulness, individual rationality, and efficiency in the multi-round game. Similar to \cref{th:approximate}, the analysis builds on the properties of the VCG mechanism and the learning error bound from \cref{lem:epsilon}. However, since we now consider an online setting, the properties need to be analyzed both for exploration and exploitation phases.

We write $\mathcal T_\text{explore}$ for the set of exploration rounds and $\mathcal T_\text{exploit}$ for the set of exploitation rounds. The key idea to achieve sublinear welfare regret is to balance exploration length $K$ and exploitation length $M_s$. In our algorithm, which separates exploration and exploitation, the exploitation length $M_s$ must progressively become longer as estimates improve. We use the following lemma for our proofs.
	\begin{lemma}[\citet{kandasamy2023vcg}, Lemma 18]
    \label{lem:stage}
	Let $M_s=\lfloor\frac{5}{6} K\sqrt{s_t}\rfloor$ for stage $s_t$ in round $t$. Then,
    \[
		\begin{aligned}
		s_t\leq 3 K^{-2/3}t^{2/3},&\quad \text{if } t\in\mathcal T_{explore},\\
	\frac{1}{2}K^{-2/3}t^{2/3}\leq s_t,&\quad \text{if } t\in\mathcal T_{exploit}.
	\end{aligned}
    \]
	\end{lemma}
	Each stage consists of an exploration phase followed by an exploitation phase. Hence, the number of queries $Q_t$ up to round $t$ is given by:
	\begin{align}
	Q_t\leq Ks_t,&\quad \text{if } t\in\mathcal T_{explore},
    \label{eq:queries-exploration}\\
	Q_t=Ks_t,&\quad \text{if } t\in\mathcal T_{exploit}.
    \label{eq:queries-exploitation}
	\end{align}
	Now, we can control the learning error for pairwise preferences using $\epsilon_{Q_t}(\delta)$.
	\paragraph{Truthfulness.} We want to bound the cumulative utility gain when agent $i$ is untruthful. In exploration rounds, utilities are independent of agents' feedback, because allocations are selected based on a D-optimal design and payments are constant. For exploitation rounds, we have:

\[
 U_i(T)-\bar U_i(T)=\sum_{t\in\mathcal T_\text{exploit}} \left( u_i(\hat a_t(\D_i^{s_t},\D_{-i}^{s_t}))-u_i(\hat a_t(\bar\D_i^{s_t},\D_{-i}^{s_t}))\right),
\]
We apply the one-shot truthfulness bound \cref{eq:oneshot-truthful} and get
\[U_i(T)-\bar U_i(T)\leq \sum_{t\in\mathcal T_\text{exploit}}\epsilon_{Q_t}(\delta),\]
where $Q_t$ is the number of queries up to round $t$.
As $\epsilon_{Q_t}(\delta)\geq 0$, we extend the sum to all rounds such that
\begin{equation}
	U_i(T)-\bar U_i(T)\leq \sum_{t=1}^T \epsilon_{Q_t}(\delta)=\sum_{t=1}^T \gamma_t(\delta)\sqrt{\frac{d}{Q_t-d(d+1)/2}}.
	\label{eq:sum-epsilon}
\end{equation}
Using Equation \cref{eq:queries-exploitation} with \cref{lem:stage} to lower bound $Q_t$, we get  with probability at least $1-\delta$:
\begin{equation}
	U_i(T)-\bar U_i(T)\lesssim\sum_{t=1}^T \gamma_t(\delta)\sqrt{2d}K^{-1/6}t^{-1/3}\leq \frac{3}{2}\gamma_T(\delta)\sqrt{2d}K^{-1/6}T^{2/3}.
	\label{eq:multi-truthfulness}
\end{equation}
Here, $\lesssim$ indicates that we dropped the additive rounding term $d(d+1)/2$ in the denominator, since it does not affect asymptotic behavior. Then, we used monotonicity of $\gamma_t(\delta)$ \cref{eq:beta} and the identity $\sum_{t=1}^T t^{-1/3}\leq \frac{3}{2}T^{2/3}$.
\paragraph{Individual rationality.} Next, we consider a truthful agent's utility under allocations $\hat a_t:=\hat a_t (\bar \D_i^{s_t},\D_{-i}^{s_t})$. In exploration rounds, the mechanism makes constant payments $p_i(\hat a_t)=c_{\max}$ for each queried allocation. We set $c_{\max}=BL$ under \cref{as:linear-cost}, then agents have utility
	\[u_i(a_t)=p_i(\hat a_t)-c_i(\hat a_t)=c_{\max}-c_i(\hat a_t)\geq 0.\]

For exploitation rounds, we apply the one-shot individual rationality bound \cref{eq:oneshot-ir} to obtain
\[\bar U_i(T)=\sum_{t\in\mathcal T_\text{exploit}}u_i(\hat a_t )\geq\sum_{t\in\mathcal T_\text{exploit}} -\epsilon_{Q_t}(\delta).
\]
Note that this bound uses \cref{as:zero-allocation}. By the same arguments as in \cref{eq:sum-epsilon} and \cref{eq:multi-truthfulness}, we get with probability at least $1-\delta$: 
\[\bar U_i(T)\gtrsim\sum_{t=1}^T -\epsilon_{Q_t}(\delta)\geq-\frac{3}{2}\gamma_T(\delta)\sqrt{2d}K^{-1/6}T^{2/3}.\]
\paragraph{Efficiency.}
Finally, we want to upper bound the welfare regret \cref{eq:welfare-regret} when all agents provide truthful feedback. We decompose the total regret as
\[R^w(T)=\sum_{t=1}^T r_t=\sum_{t\in\mathcal T_\text{explore}} r_t+\sum_{t\in\mathcal T_\text{exploit}} r_t,\]
where $r_t=J(\hat a_t(\bar\D^{s_t}))-J(a^*)$ is the instantaneous regret at time $t$. In exploration rounds, the social cost for a pairwise comparison is at most $2\,J_{\max}$, with an upper bound $J_{\max}=NBL$ under Assumption \ref{as:linear-cost}. Then, we have
\[
\begin{split}
\sum_{t\in\mathcal T_\text{explore}} r_t&\overset{(i)}{\leq} 2J_{\max}Q_T\\&\overset{(ii)}{\leq} 6NBLK^{1/3}T^{2/3}.   
\end{split}
\]
In (i), we used that $r_i\leq2J_{\max}$ for $Q_T$ total queries, and (ii) used Equation \cref{eq:queries-exploration} with \cref{lem:stage} to upper bound $Q_T$. For exploitation rounds, we apply the one-shot efficiency bound \cref{eq:oneshot-efficient}. By the same arguments as in \cref{eq:sum-epsilon} and \cref{eq:multi-truthfulness}, and taking a union bound over $N$ agents, we get with probability at least $1-\delta$:
\[\sum_{t\in\mathcal T_\text{exploit}} r_t\lesssim \sum_{t=1}^T N\epsilon_{Q_t}(\delta/N)\leq \frac{3}{2}N\gamma_T(\delta/N)\sqrt{2d}K^{-1/6}T^{2/3}.\]
Combining the bounds yields with probability at least $1-\delta$:
\[R^w(T)\lesssim 6NBLK^{1/3}T^{2/3}+\frac{3}{2}N\gamma_T(\delta/N)\sqrt{2d}K^{-1/6}T^{2/3}.\]
\end{proof}
\section{Unknown rationality parameter}
\label{app:rationality}
We extend the Bradley--Terry model \cref{eq:bradley-terry} by introducing an agent-specific rationality parameter $\beta_i>0$, also known as the inverse temperature. The probability that agent $i$ prefers allocation $a\in\mathcal A$ over allocation $a'\in\mathcal A$ is then given by:
\[
\begin{aligned}
\Pr(a \succ a') &=\sigma\bigl(\beta_i[u_i(a)-u_i(a')]\bigr)\\&=\sigma\bigl(\beta_i[p_i(a)-c_i(a)-p_i(a')+c_i(a')]\bigr),
\end{aligned}
\]
where $\sigma(x)=1/(1+e^{-x})$ denotes the sigmoid function. Under \cref{as:linear-cost} and defining $\Delta p_i:=p_i(a)-p_i(a')$, we equivalently write:
\begin{equation}
\label{eq:temperature}
  \Pr(a\succ a')=\sigma\bigl(\langle \beta_i \theta_i^*,\phi_i(a')-\phi_i(a)\rangle+\beta_i\Delta p_i\bigr).  
\end{equation}
In the main paper and the subsequent analysis, we assumed $\beta_i$ to be known and, without loss of generality, set $\beta_i=1$. This assumption in the context of agent-specific $\beta_i>0$ can also be interpreted as agents with different levels of rationality who adjust their thinking times, so that the central planner can learn $\theta_i^*$ \citep{strzalecki2025stochastic}.

In practice, however, such an adjustment may be difficult for agents. To learn $\theta_i^*$, the central planner then also needs to estimate $\beta_i$. To this end, we outline a simple subroutine for the mechanism to learn $\beta_i$ in the one-shot game. Following the earlier procedure in \cref{alg:oneshot}, the mechanism first computes a scaled cost parameter $\tilde \theta_i\approx\beta_i\theta_i^*$. Now, the subroutine queries preferences over two allocations $a,a'\in\mathcal A$ and makes payments with $\Delta p_i\neq0$.  To compute $\hat\beta_i$, in principle any pair of allocations could be used, as all quantities except $\beta_i$ in Equation \cref{eq:temperature} are known or have been estimated. To simplify the analysis, we consider the case $a=a'$, which removes the dependence on $\tilde\theta_i$. Then, the agent's truthful preference is expressed by
\[
    \bar y_i\sim\text{Bernoulli}\bigl(\sigma(\beta_i\Delta p_i)\bigr).
\]
From these preferences, the social planner can estimate $\hat\beta_i$ via maximum likelihood. By \cref{lem:schlaginhaufen}, we can find high-probability bounds after $K_1$ and $K_2$ queries, respectively, such that:
\begin{align}
\Pr(\|\tilde\theta_i-\beta_i\theta_i^*\|_{V_{i,K_1}}\leq\gamma_1)\geq 1-\delta_1,
\label{eq:gamma-scaled}
\\
\Pr(|\hat\beta_i-\beta_i|\leq\gamma_2)\geq 1-\delta_2,
\label{eq:gamma-beta}
\end{align}
where $\gamma_1\in\mathcal O(\log(K_1))$ and $\gamma_2\in\mathcal O(\log(K_2)K_2^{-1/2})$. The estimation error of the scaled parameter $\hat\theta_i$ is
\[
\begin{split}
  \hat\theta_i-\theta_i^*&\overset{(i)}{=}\frac{\tilde\theta_i}{\hat\beta_i}-\frac{\beta_i\theta_i^*}{\beta_i}\\&\overset{(ii)}{=}\frac{\tilde\theta_i-\beta_i\theta_i^*}{\hat\beta_i}+\beta_i\theta_i^*\left(\frac{1}{\hat\beta_i}-\frac{1}{\beta_i}\right).  
\end{split}\]
In (i), we used that $\hat\theta_i=\tilde\theta_i/\hat\beta_i$ and (ii) added and subtracted $\beta_i\theta_i^*/\hat\beta_i$. Plugging this term in \cref{eq:linear-difference}, we get
\[
\begin{split}
\left|(c_i(a) - c_i(a')) - \left(\hat{c}_i(a; \bar{\mathcal D}_i)) -  \hat{c}_i(a'; \bar{\mathcal D}_i)\right)  \right|&=\left|\langle\hat\theta_i-\theta_i^*,\phi_i(a')-\phi_i(a)\rangle\right|\\&=\left|\left\langle\frac{\tilde\theta_i-\beta_i\theta_i^*}{\hat\beta_i}+\beta_i\theta_i^*\left(\frac{1}{\hat\beta_i}-\frac{1}{\beta_i}\right),\phi_i(a')-\phi_i(a)\right\rangle \right|.
\end{split}
\]

Applying the triangle inequality gives
\[
\begin{split}
\left|(c_i(a) - c_i(a')) - \left(\hat{c}_i(a; \bar{\mathcal D}_i))-\hat{c}_i(a'; \bar{\mathcal D}_i)\right)  \right|&\leq\left|\left\langle\frac{\tilde\theta_i-\beta_i\theta_i^*}{\hat\beta_i},\phi_i(a')-\phi_i(a)\right\rangle \right|\\&\qquad+\left|\left\langle\beta_i\theta_i^*\left(\frac{1}{\hat\beta_i}-\frac{1}{\beta_i}\right),\phi_i(a')-\phi_i(a)\right\rangle\right|.
\end{split}
\]
For the first term, we apply the Cauchy-Schwarz inequality and get with probability at least $1-\delta_1$:
\[
\left|\left\langle\frac{\tilde\theta_i-\beta_i\theta_i^*}{\hat\beta_i},\phi_i(a')-\phi_i(a)\right\rangle \right|\leq\frac{\|\tilde\theta_i-\beta_i\theta_i^*\|_{V_{i,K_1}}}{|\hat\beta_i|}\|\phi_i(a')-\phi_i(a)\|_{V_{i,K_1}^{-1}}\leq\frac{\gamma_1}{|\hat\beta_i|}\sqrt{\frac{d}{K_1-d(d+1)/2}},
\]
using the bound \cref{eq:gamma-scaled} and \cref{lem:design}. For the second term, we have with probability at least $1-\delta_2$:
\[
\left|\left\langle\beta_i\theta_i^*\left(\frac{1}{\hat\beta_i}-\frac{1}{\beta_i}\right),\phi_i(a')-\phi_i(a)\right\rangle\right|=\frac{|\hat \beta_i-\beta_i|}{|\hat\beta_i|}|\langle\theta_i^*,\phi_i(a')-\phi_i(a)\rangle|\leq\frac{\gamma_2}{|\hat\beta_i|}2BL,
\]
using the bound \cref{eq:gamma-beta} and \cref{as:linear-cost}.
Finally, combining the bounds and setting $K_1=K_2=K$, we get with probability at least $1-\delta_1-\delta_2$:
\[\left|(c_i(a) - c_i(a')) - \left(\hat{c}_i(a; \bar{\mathcal D}_i))-\hat{c}_i(a'; \bar{\mathcal D}_i)\right)  \right|\leq\frac{2\gamma_1}{\beta_i}\sqrt{\frac{d}{K-d(d+1)/2}}+\frac{4BL\gamma_2}{\beta_i} = \tilde{\mathcal{O}}\left( \dfrac{\epsilon_K(\delta_1 + \delta_2)}{\beta_i} \right),\]
where we assumed that $\gamma_2\leq\beta_i/2$ to lower bound $\hat\beta_i\geq\beta_i/2$.
Overall, this discussion shows that, given that agents are somewhat rational, unknown rationality parameters do not affect the properties of the mechanism.

\section{Experimental Configurations}
\label{app:experiments}
\cref{tab:configs} provides the configurations used for experiments in \cref{sec:experiments}. All experiments were run on a MacBook Pro 2023 with an M2 Pro chip and 16\;GB of RAM.

\begin{table}[h]
\centering
\caption{Overview of experiment parameters and runtimes.}
\begin{tabular}{lccc}
&\textbf{Truthfulness} & \textbf{Social Cost Gap} & \textbf{Welfare Regret}\\
& \cref{fig:truthfulness} & \cref{fig:efficiency-one-shot} & \cref{fig:efficiency-multi-round}\\
\hline \\
Parameters & \texttt{K = 10,110,\dots,1010} & \texttt{K = 10,110,\dots,1910} & \texttt{K = 5, T = 2000}\\
Repetitions & \texttt{20} & \texttt{50} & \texttt{20}\\
Runtime & \texttt{1:36\;h} & \texttt{0:23\;h} & \texttt{3:46\;h}\\
\end{tabular}
\label{tab:configs}
\end{table}\end{document}